\documentclass[conference]{IEEEtran}
\usepackage{amssymb,amsmath,epsfig}
\usepackage{graphicx}

\newcommand{\eqdef}{ := }
\usepackage{amsthm}
\newtheorem{theorem}{Theorem}

 \usepackage{color}
 \usepackage{subfig,float,stfloats}
\usepackage{amssymb, amsmath, epsfig, booktabs, amsthm, epstopdf, color, import, cite, bm}
\usepackage{psfrag}
\usepackage{dsfont}
\usepackage{cite}
\usepackage{bm}

\newcommand{\egd}{{e.g.}, }
\newcommand{\ied}{{i.e.}, }

\ifCLASSINFOpdf
\else
\fi
\allowdisplaybreaks[4]

\renewcommand{\baselinestretch}{0.96}
\IEEEoverridecommandlockouts
\begin{document} 
\title {\huge On Downlink Interference Decoding In Multi-Cell Massive MIMO Systems \vspace*{-4mm}} 
\author{Meysam~Shahrbaf~Motlagh\thanks{\noindent\textsuperscript{$\dagger$}Meysam Shahrbaf Motlagh and Patrick Mitran are with the University of Waterloo, ON, Canada.}\textsuperscript{$\dagger$},~\textrm{Subhajit~Majhi\thanks{\noindent\textsuperscript{$\ast$}Subhajit Majhi contributed to this work while he was affiliated with the University of Waterloo, ON, Canada, and is currently with Qualcomm India.}}\textsuperscript{$\ast$},~\textrm{Patrick~Mitran}\textsuperscript{$\dagger$}~\textrm{and}~\textrm{Hideki~Ochiai\thanks{\noindent\textsuperscript{$\ddagger$}Hideki Ochiai is with Yokohama National University, Yokohama, Japan.}}\textsuperscript{$\ddagger$} 
{\thanks{This work was supported in part by the Natural Sciences and Engineering Research Council of Canada and Cisco.}}
} 
\IEEEaftertitletext{\vspace{-7mm}}
 \maketitle 
 		\bstctlcite{IEEEexample:BSTcontrol}%
 
\begin{abstract}
In this paper, the downlink of a multi-cell massive MIMO system is considered where the channel state information (CSI) is estimated via pilot symbols that are orthogonal in a cell but re-used in other cells. Re-using the pilots, however, contaminates the CSI estimate at each base station (BS) by the channel of the users sharing the same pilot in other cells. The resulting inter-cell interference does not vanish even when the number of  BS antennas $M$ is large, \ied $M\rightarrow\infty$, and thus the rates achieved by treating interference as noise (TIN) saturate even if $M\rightarrow\infty$. In this paper, interference aware decoding schemes based on simultaneous unique decoding (SD) and simultaneous non-unique decoding (SND) of the full interference or a part of the interference (PD) are studied with two different linear precoding techniques: maximum ratio transmission (MRT) and zero forcing (ZF). The resulting rates are shown to grow unbounded as $M\rightarrow\infty$. In addition, the rates achievable via SD/SND/PD for finite $M$ are derived using a worst-case uncorrelated noise technique, which are shown to scale as $\mathcal{O}(\log M)$. To compare the performance of different schemes, the maximum symmetric rate problem is studied, where it is confirmed that with large, yet practical, values of $M$, SND strictly outperforms TIN, and also that PD strictly outperforms SND. 
\end{abstract} 
 \IEEEpeerreviewmaketitle
\vspace{-1mm}\section{Introduction} \vspace{-1mm}
Massive MIMO systems with base stations (BSs) simultaneously serving several users via a large number of antennas, are considered to be an integral part of 5G standards \cite{bjornson2019massive}. Moreover, the benefits of massive MIMO systems reported in \cite{bjornson2019massive}, \egd high spectral efficiency, increased capacity due to aggressive spatial multiplexing and high energy efficiency can all be achieved with simple linear precoding techniques.

The use of time-division duplex (TDD) mode has been suggested as the ``canonical form'' of massive MIMO \cite{bjornson2019massive}, where downlink and uplink channels are assumed to be reciprocal. When using TDD mode, channel estimates are obtained at BSs using uplink training pilots, which are then employed for precoding/decoding. However, as the length of the channel coherence time is finite, the sequences of orthogonal pilots are limited in number. Hence, in a multi-cell system, one way to address this limitation is to re-use a set of orthogonal pilots across all cells. Thus, the channel estimate of a user in one cell is contaminated by the channels of users in other cells that use the same pilot sequence as the former \cite{motlagh2019performance}, which is referred to as pilot contamination (PC).

Interestingly, when the channels of users to BSs in a massive MIMO system are assumed to be independent Rayleigh fading, as the number of BS antennas $M$ grows, the effects of multi-user interference, additive noise, and small-scale fading all vanish owing to channel hardening and asymptotic favorable propagation \cite{marzetta2010noncooperative}. However, inter-cell interference due to PC does not vanish \cite{marzetta2010noncooperative}, which constitutes a bottleneck in multi-cell massive MIMO systems. Extensive works have been carried out in the literature, such as those in \cite{yin2016robust, adhikary2017uplink, bjornson2018massive}, that have proposed solutions to tackle the PC problem, but under some constraints that may be difficult to realize in practice. Specifically, these include among others, the assumptions of (i) channel coherence time approaching infinity \cite{yin2016robust}, (ii) the presence of a central unit and cumbersome backhaul overhead to share all data symbols for centralized processing \cite{adhikary2017uplink}, (iii) channels of pilot-sharing users in different cells with asymptotically linearly independent covariance matrices \cite{bjornson2018massive}.

Most studies in the literature treat the PC interference as noise (TIN) and the resulting rates, which are suboptimal in many cases \cite{bandemer2011interference}, saturate and do not grow with $M$. In this paper, we address the PC interference problem in the downlink of a multi-cell massive MIMO system arising from re-using the same set of orthogonal pilots across cells. In particular, we depart from the viewpoint of treating the PC interference as noise and propose full/partial interference decoding schemes that attain unbounded rates for large $M$.  

The major contributions of this paper are summarized below. First, for two well-known linear precoding techniques at the BS, maximum ratio transmission (MRT) and  zero forcing (ZF), we derive achievable rates for schemes that decode the PC interference either fully such as the simultaneous unique/non-unique decoding (SD/SND), or partially (PD). Specifically, the PC interference is fully decoded uniquely in SD, and non-uniquely in SND. We similarly study achievable rates for a PD scheme, where messages at each BS are encoded into two independent layers; then only one layer of PC interference is decoded non-uniquely while the other layer is treated as noise. The per-user achievable rates for SD/SND/PD, derived using a worst-case uncorrelated noise technique, are shown to scale as $\mathcal{O} (\log M)$, thus eliminating the rate saturation problem. Focusing on the 2-cell case, we study the problem of maximum symmetric rate, \ied  maximizing the worst user's rate, for TIN/SD/SND/PD schemes, and observe that the full interference decoding scheme of SND outperforms TIN, while the partial decoding scheme of PD outperforms even SND.

Note that in \cite{motlagh2019performance}, full PC interference decoding schemes were proposed for the uplink of a multi-cell massive MIMO system using maximum ratio combining (MRC). In contrast, the current paper focuses on the downlink setting with both MRT and ZF and proposes not only full (SD/ SND) but also partial interference decoding (PD) schemes based on rate splitting. In particular, we find that the use of ZF which better mitigates multi-user interference in conjunction with either SND or PD results in notably larger rates such that the number of BS antennas required for interference decoding schemes to outperform TIN now reduces by more than a factor of $\approx 100$ (i.e., within the practical range), which is a significant improvement over the results of \cite{motlagh2019performance}.

\vspace*{-1mm}
\section{Preliminaries}\label{sec:2} \vspace*{0mm}
\subsection{System Model}\vspace*{0mm}
 A TDD multi-cell system with $L$ cells is considered, where each cell comprises a BS that serves $K$ single antenna users with $M$ transmitting antennas ($M \gg K$). The channel between the BS in cell $j$ ($\mathsf{BS}_j$) and users in cell $l$ is denoted as $\pmb{G}_{jl}=\pmb{H}_{jl}\pmb{D}_{jl}^{1/2},
$ where $\pmb{H}_{jl}=[\pmb{h}_{j1l},\pmb{h}_{j2l}, ..., \pmb{h}_{jKl}] \in \mathbb{C}^{M\times K}$ is the channel matrix corresponding to the channel vectors $\pmb{h}_{jkl} \in \mathbb{C}^{M \times 1}$ of small-scale fading coefficients between $\mathsf{BS}_j$ and the $k^{th}$ user in cell $l$ ($\mathsf{U}_{kl}$), and $\pmb{D}_{jl} = \textrm{diag}(\beta_{j1l}, \beta_{j2l}, ..., \beta_{jKl} )$ is the matrix containing the large-scale fading coefficients. Assuming a block fading model, the large-scale fading coefficient $\beta_{jkl}$, which models path loss and shadowing between $\mathsf{BS}_j$ and $\mathsf{U}_{kl}$, is taken to be a constant that is known \emph{a priori} at $\mathsf{BS}_j$, remains unchanged over multiple coherence time-intervals, and is independent of the antenna index of $\mathsf{BS}_j$. The small-scale fading coefficients $\pmb{h}_{jkl} \sim \mathcal{CN} (\pmb{0}, \, \pmb{I}_M)$ are also independent across coherence intervals but remain constant in a given coherence interval, with $\pmb{I}_M$ denoting the identity matrix of order $M$. The channel between $\mathsf{BS}_j$ and $\mathsf{U}_{kl}$ can also be expressed as $\pmb{g}_{jkl}= \sqrt{\beta_{jkl}} \pmb{h}_{jkl}.$ Finally, the downlink and uplink channels are assumed reciprocal.

\emph{Downlink Data Transmission:} The baseband signal received at $\mathsf{U}_{il}$ is given by \setlength{\abovedisplayskip}{-2pt}\setlength{\belowdisplayskip}{1pt}\vspace{2mm}
\begin{equation}\label{eq:3}
y_{il} = \sum\nolimits_{j=1}^L \sqrt{\rho_{\rm d}} \pmb{g}_{jil}^{\dag} \pmb{x}_j + w_{il} ,\vspace{2mm}
\end{equation}
where $\pmb{x}_j = \left[ x_j[1], x_j[2], ..., x_j[M] \right]^T $ is the signal transmitted  by $\mathsf{BS}_j$, $w_{il} \sim \mathcal{CN} (0,  1)$ is the AWGN at $\mathsf{U}_{il}$, $(.)^\dagger$ denotes conjugate transpose, and $\rho_{\rm d}$ can be interpreted as the \emph{downlink} transmission SNR per user of the BS. Defining $\pmb{u}_{jk} \in \mathbb{C}^{M \times 1}$ to be the precoding vector of $\mathsf{BS}_j$, we have $\textstyle \pmb{x}_j = {1}/{\sqrt{\lambda_j}}\sum_{k=1}^K \pmb{u}_{jk} s_j[k] = {\pmb{U}_j \pmb{s}_j}/{\sqrt{\lambda_j}}$ where $\pmb{s}_j = \left[ s_j[1], s_j[2], ..., s_j[K] \right]^T $ is the vector of data symbols intended for the users in cell $j$, $\pmb{U}_j = \left[ \pmb{u}_{j1}, \pmb{u}_{j2}, ..., \pmb{u}_{jK} \right] \in \mathbb{C}^{M \times K}$, and $\lambda_j$ is a normalization constant such that the downlink SNR per user for $\mathsf{BS}_j$ is $\rho_{\rm d}$, \ied ${\mathbb{E}[\rho_{\rm d} \pmb{x}_j^{\dag} \pmb{x}_j ]}/{ K} = \rho_{\rm d}$.

\emph{CSI Estimation at BS}: Similar to \cite{adhikary2017uplink}, a set of orthogonal pilot sequences is re-used in all cells, and thus the channel estimate will be contaminated with the PC interference from users in other cells that are using the same pilot. During the uplink training phase, $\mathsf{U}_{kj}$ transmits a pilot to $\mathsf{BS}_j$, $k=1,\ldots, K, j = 1,\ldots,L$, and then  $\mathsf{BS}_j$ estimates the local channels $\pmb{G}_{jj}$, given by $\pmb{\hat{G}}_{jj}$. 

The minimum mean squared error (MMSE) estimate $\pmb{\hat{g}}_{jkj}$ of $\pmb{g}_{jkj}$ can be shown to be \cite{adhikary2017uplink}   
\setlength{\abovedisplayskip}{1pt}\setlength{\belowdisplayskip}{1pt}\vspace{2mm}
\begin{align}\label{eq:6}
\pmb{\hat{g}}_{jkj} = \alpha_{jkj} \left( \sum\nolimits_{l=1}^L \sqrt{\rho_{\textrm{p}}} \pmb{g}_{jkl} + \pmb{\tilde{z}}_{jk} \right),
\end{align} 
where $\alpha_{jkj} := {\sqrt{\rho_{\textrm{p}}} \beta_{jkj}}/({1 + \rho_{\textrm{p}}\sum_{l=1}^L \beta_{jkl}})$,  $\pmb{\tilde{z}}_{jk} \sim \mathcal{CN} (\pmb{0}, \; \pmb{I}_M)$, and $\rho_{\textrm{p}}$ is the pilot SNR. Also note that the channel vector $\pmb{g}_{jkj}$ can be decomposed as $\pmb{g}_{jkj} = \pmb{\hat{g}}_{jkj} + \pmb{\epsilon}_{jkj}$, where the estimation error $\pmb{\epsilon}_{jkj} \sim \mathcal{CN} \left( \pmb{0}, \; \beta_{jkj} \left( 1 - \sqrt{\rho_{\textrm{p}}} \alpha_{jkj} \right) \pmb{I}_M \right) $ is independent of the MMSE estimate $\pmb{\hat{g}}_{jkj} \sim \mathcal{CN} \left( \pmb{0}, \; \sqrt{\rho_{\textrm{p}}} \beta_{jkj} \alpha_{jkj} \pmb{I}_M \right)$ \cite{adhikary2017uplink}. 

\emph{Treating Interference as Noise (TIN)}:
For MRT at $\mathsf{BS}_j$, the precoding vector is $\pmb{u}_{jk}=\pmb{\hat{g}}_{jkj}$, and the normalization factor $\lambda_j^{\rm mrt}$, as found from the condition ${\mathbb{E}[\rho_{\rm d} \pmb{x}_j^{{\dag}} \pmb{x}_j ]}/{ K} \!=\! \rho_{\rm d}$, is $\lambda_j^{\rm mrt} \!=\! \frac{M}{K} \sum_{k=1}^K\!\sqrt{\rho_{\textrm{p}}} \beta_{jkj} \alpha_{jkj}$.
Thus, rewriting \eqref{eq:3} we have \vspace{2mm}
\begin{equation}\label{eq:11}
y_{il}^{\rm mrt} = \sum\nolimits_{j=1}^L \sqrt{{\rho_{\rm d}}/{\lambda_j^{\rm mrt}}} \sum\nolimits_{k=1}^K \pmb{g}_{jil}^{\dag} \pmb{\hat{g}}_{jkj} s_j [k] + w_{il}. \vspace{2mm}
\end{equation}
Using \eqref{eq:6} and applying the strong law of large numbers to \eqref{eq:11}, we obtain\vspace{2mm}
\begin{equation}\label{eq:12}
\footnotesize
\footnotesize \dfrac{y_{il}^{\textrm{\rm mrt}}}{\sqrt{M}} \stackrel{\textrm{a.s.}}{\rightarrow}  \left( \frac{\sqrt{K \rho_{\textrm{d}}\rho_{\textrm{p}}} \beta_{lil} \alpha_{lil} s_l [i] }{\sqrt{\sum_{k=1}^K \sqrt{\rho_{\rm p}} \beta_{lkl} \alpha_{lkl}  }}  + \hspace{-2mm}\sum_{j=1, j \neq l}^L \frac{ \sqrt{K \rho_{\textrm{d}}\rho_{\textrm{p}}} \beta_{jil} \alpha_{jij} s_j [i] }{\sqrt{\sum_{k=1}^K \sqrt{\rho_{\rm p}} \beta_{jkj} \alpha_{jkj} }}  \right),\vspace{2mm}
\end{equation} 
as $M \rightarrow \infty$, \ied a channel hardening effect is observed.
 
Assuming TIN in the downlink, $\mathsf{U}_{il}$ decodes only the desired signal $s_l [i]$ while treating the interfering signals $s_{j} [i], \; j \neq l$ as noise. Hence, defining $R_{il}$ as the downlink rate associated with decoding $s_l [i]$, any rate tuple $\left( R_{i1}, ..., R_{iL} \right)$ is achievable if the following conditions hold\vspace{2mm}
\begin{equation}\label{eq:12:3} 
R_{il} \leq I \left( y_{il}^{\rm mrt}  ; s_l [i] \right),   \vspace{2mm} 
\end{equation}
for $l=1,\ldots, L, \; i=1,\ldots, K$, where $I(X;Y)$ is the mutual information between $X$ and $Y$ \cite{el2011network}. 
Using \eqref{eq:12}, it is verified that the right-hand side (r.h.s) in \eqref{eq:12:3} converges to a constant independent of $M$.
A similar phenomenon is observed when ZF precoding is used at the BS. In particular, considering ZF precoding matrix at $\mathsf{BS}_j$ given by $\pmb{V}_{jj} = \left[ \pmb{v}_{j1j}, ..., \pmb{v}_{jKj} \right] = \pmb{\hat{G}}_{jj} \left( \pmb{\hat{G}}_{jj}^{\dagger} \pmb{\hat{G}}_{jj} \right)^{-1}$, it can be verified that $\pmb{V}_{jj}^{\dagger} \pmb{\hat{G}}_{jj} = \pmb{I}_K$, and thereby $\pmb{v}_{jkj}^{\dagger} \pmb{\hat{g}}_{jmj} = \delta_{mk}$, where $\delta_{mk}$ is the Kronecker delta function. The normalization factor $\lambda_j^{\rm zf}$ is also derived to be $\lambda_j^{\rm zf} = \frac{1}{K (M-K)} \sum_{k=1}^K 1/ (\sqrt{\rho_{\rm p}} \beta_{jkj} \alpha_{jkj} ),$ from which an analysis similar to that of TIN in MRT results in the saturation of the rate obtained by ZF, i.e., $\!I \!\left( y_{il}^{\textrm{zf}} ; s_l [i]  \right)$ converges to a constant almost surely, as $M \rightarrow \infty$. 
Hence, the benefits of increasing $M$ saturate since the TIN rates obtained by MRT/ZF converge to a constant as $M$ becomes large. 
 \section{Decoding the PC Interference in Downlink} \vspace{0mm}
We propose schemes that treat the PC interference terms as individual users, and thus decode them fully or partially, as opposed to treating them as noise. This results in new achievable rates that grow unbounded with $M$.

\vspace{0mm}\subsection{Simultaneous Unique Decoding (SD)}\vspace{0mm}
In SD, the PC interference terms are decoded uniquely, details of which for MRT and ZF are given below.

\emph{MRT:}  Note that in \eqref{eq:12}, which expresses the received signal after performing MRT and letting $M \rightarrow \infty$,  the first term is the desired signal and the other non-vanishing terms are the inter-cell PC interference. Treating PC interference terms as users to be uniquely decoded, \eqref{eq:12} becomes identical to the output of a noise-free multiple-access channel (MAC). More specifically, after performing MRT at the BSs and letting $M \rightarrow \infty$, a noise-free $L$-user MAC is obtained at the user's side in downlink: for this MAC, unbounded per-user rates can be obtained by uniquely jointly decoding all signals $\lbrace s_l[i] \rbrace_{l=1}^L $.

We provide an analysis that establishes the achievable rates for finite $M$. We add and subtract a term associated with the mean of the effective channel $\pmb{g}_{jil}^{\dag} \pmb{\hat{g}}_{jij}$ in \eqref{eq:11}, which results in the following  \par\nobreak \vspace*{0mm} \begingroup
\addtolength{\jot}{0mm} \small{
	\begin{align} 
	\nonumber y_{il}^{\rm mrt} &= \underbrace{ \sum\nolimits_{j=1}^{L} \sqrt{\dfrac{\rho_{\rm d}}{\lambda_j^{\rm mrt}}} \mathbb{E} \left[ \pmb{g}_{jil}^{\dag} \pmb{\hat{g}}_{jij} \right] s_j [i]}_{\textrm{Desired signals}} \\
	\nonumber &\hspace{5mm}+ \underbrace{ \sum\nolimits_{j=1}^L \sqrt{\dfrac{\rho_{\rm d}}{\lambda_j^{\rm mrt}}}\left( \pmb{g}_{jil}^{\dag} \pmb{\hat{g}}_{jij} - \mathbb{E} \left[ \pmb{g}_{jil}^{\dag} \pmb{\hat{g}}_{jij} \right] \right) s_j [i] }_{\textrm{Interference due to beamforming gain uncertainty}}   \\
	&\hspace{5mm} + \underbrace{ \sum\nolimits_{j=1}^L \sqrt{\dfrac{\rho_{\rm d}}{\lambda_j^{\rm mrt}}} \sum\nolimits_{k=1, k \neq i}^K \pmb{g}_{jil}^{\dag} \pmb{\hat{g}}_{jkj} s_j [k]}_{\textrm{Inter-cell interference}} + \underbrace{ w_{il}}_{\textrm{Noise}} \label{eq:18} \\
	&= \sum\nolimits_{j=1}^L  \theta_{jij}^{\textrm{mrt}} s_j [i] + w_{il}^{\prime}, \label{eq:19}
	\end{align}}%
\endgroup where $\theta_{jij}^{\textrm{mrt}} = \sqrt{{\rho_{\rm d}}/{\lambda_j^{\rm mrt}}} \mathbb{E} \left[ \pmb{g}_{jil}^{\dag} \pmb{\hat{g}}_{jij} \right] $ and $w_{il}^{\prime}$ is the effective noise, incorporating the last three terms in \eqref{eq:18}.

Consider the $L$-user MAC in \eqref{eq:19} associated with the estimate of $s_l [i]$ at $\mathsf{U}_{il}$. Using the usual definitions of  \cite{el2011network}, messages $m_{il} \in \left[ 1 : 2^{n R_{il}} \right]\!,  l=1, ..., L $ (distributed uniformly) are communicated over this MAC by encoding them into codewords $\pmb{s}_l[i](m_{il})$ of length $n$, where ${s}_l[i] \sim \mathcal{CN} (0, 1)$, i.i.d. Using the standard random coding analysis \cite{el2011network} and decoding all messages uniquely as in SD, decoding error probability can be shown to approach zero as $n \rightarrow \infty$,  if \setlength{\abovedisplayskip}{1pt}\setlength{\belowdisplayskip}{1pt}
\begin{align}\label{eq:16}
 \sum\nolimits_{l \in \Omega} R_{il} \leq I \left(  y_{il}^{\rm mrt} ; \; \pmb{s}_{\Omega} \; \Big\vert \; \pmb{s}_{\Omega^c} \right), \vspace{2mm}
\end{align}
for all $\Omega \subseteq S= \lbrace 1, 2,\ldots, L\rbrace$, and $\pmb{s}_{\Omega}$ is the vector with entries $ s_l[i],\; l \in \Omega$. Thus, the rate tuple $\left( R_{i1}, ..., R_{iL} \right)$ in \eqref{eq:16} is achievable. Finally, the network-wide achievable region is obtained by taking the intersection of the achievable regions over all receivers.

As opposed to an ordinary MAC, where noise is typically Gaussian and independent of users' signals, for the MAC in \eqref{eq:19} it is neither Gaussian nor independent of the users' signals, which makes it hard to characterize exact expressions for $I(.;.|.)$ terms in \eqref{eq:16}. However, the effective noise term in \eqref{eq:19} is uncorrelated from the desired signals. Thus, we are able to establish lower bounds for these mutual information terms based on a worst-case uncorrelated noise technique as follows. \vspace{-3mm}
\begin{theorem}\label{lemma:3}
	For $\pmb{s}_l = \left[ s_l[1], s_l[2], ..., s_l[K] \right]^T \sim \mathcal{CN} \left(\pmb{0}, \; \pmb{I}_K  \right)$, $l \in \lbrace 1, 2, ..., L \rbrace $, a set of lower bounds to \eqref{eq:16} that can be achieved with MRT for the MAC given in \eqref{eq:19}, associated with $\mathsf{U}_{il}$, is given as follows
	\begin{equation}\label{eq:21} 
	I \left(  y_{il}^{\rm mrt} ; \; \pmb{s}_{\Omega}  \Big\vert  \pmb{s}_{\Omega^c} \right) \!\geq \!C \left( { P_1^{\rm mrt} }/{\left(P_2^{\rm mrt} + P_3^{\rm mrt} + P_4^{\rm mrt}\right)} \right)\! ,\!\!
	\end{equation}  where
	\par\nobreak \vspace*{-2mm} \begingroup \addtolength{\jot}{-0.5mm} { \begin{align} 
		P_1^{\rm mrt} &=  M^2 \sum\nolimits_{j \in \Omega}    \rho_{\rm d} \rho_{\rm p}  \beta_{jil}^2  \alpha_{jij}^2 /\lambda_j^{\rm mrt}   \label{eq:22} \\
		P_2^{\rm mrt} &=  M \sum\nolimits_{j=1}^L   \rho_{\rm d} \sqrt{\rho_{\rm p}} \beta_{jij} \alpha_{jij}  \beta_{jil}/ \lambda_j^{\rm mrt}  \label{eq:22:1} \\
		P_3^{\rm mrt} &=  M \sum\limits_{j=1}^L   \sum\nolimits_{k=1, k \neq i}^K  \!  \rho_{\rm d}\sqrt{\rho_{\rm p}} \beta_{jkj} \alpha_{jkj}  \beta_{jil}  / \lambda_j^{\rm mrt},     \label{eq:22:2}  	\\
 	 P_4^{\rm mrt} &= 1. \label{eq:22:3} 
		\end{align}}%
 \endgroup
\end{theorem}    
\begin{proof}
	The proof is relegated to Appendix A.
\end{proof} \vspace{-2mm}
 The terms $P^{\rm mrt}_{(.)}$ in Theorem~\ref{lemma:3} can be interpreted as follows. $P_1^{\rm mrt}$ is the power of the desired signals associated with $\pmb{s}_{\Omega}$, whereas $P_2^{\rm mrt}$ is the power of the interference due to beamforming gain uncertainty, $P_3^{\rm mrt}$ is the power of the interference of other users, and $P_4^{\rm mrt}$ is the power of the noise. Moreover, using \eqref{eq:22}-\eqref{eq:22:3}, the r.h.s of \eqref{eq:21} is expressed as \par\nobreak \vspace{-0mm}  \addtolength{\jot}{-1mm} {\small
	\begin{align}
	\!\!\!\!C_{\rm LB}^{\rm mrt} \!= C \!\left(\! \dfrac{ M^2 \sum\nolimits_{j \in \Omega}    \rho_{\rm d} \rho_{\textrm{p}}  \beta_{jil}^2  \alpha_{jij}^2 /\lambda_j^{\rm mrt}   }{M \sum_{j=1}^L   \sum_{k=1}^K  \! \left( \rho_{\textrm{d}} \sqrt{\rho_{\textrm{p}}} \beta_{jkj} \alpha_{jkj}  \beta_{jil} \right)\!/\lambda_j^{\rm mrt} +1  } \!\right) \!\!. \!\!  \label{eq:lb:mrt}
	\end{align}}%
Since for fixed $K$ and large $M$, $\lambda_j^{\rm mrt} \propto M$, the  rates in \eqref{eq:lb:mrt}  grow as $\mathcal{O} (\log M)$. Hence, by uniquely jointly decoding signals $\lbrace s_l[i] \rbrace_{l=1}^L $, unbounded rates are obtained as $M \rightarrow \infty$.

\emph{ZF:} Similar to MRT, by performing ZF at the BS and letting $M \rightarrow \infty$, a noise-free $L$-user MAC is obtained which is similar to that in \eqref{eq:12} for MRT. We now provide an analysis for ZF with finite $M$, similar to \eqref{eq:18}-\eqref{eq:19} for MRT: \par\nobreak \vspace*{-2mm} \begingroup 
\addtolength{\jot}{-0mm}{\small 
	\begin{align}\label{eq:zf:10} 
	y_{il}^{\rm zf} &= \sum\nolimits_{j=1}^L \sqrt{\dfrac{\rho_{\rm d}}{\lambda_j^{\rm zf}}} \sum\nolimits_{k=1}^K \left( \pmb{\hat{g}}_{jil}^{\dagger} + \pmb{\epsilon}_{jil}^{\dagger} \right) \pmb{v}_{jkj} s_j [k] + w_{il} \\
	&= \sum\nolimits_{j=1}^L \sqrt{\dfrac{\rho_{\rm d}}{\lambda_j^{\rm zf}}} \left( \dfrac{\beta_{jil}}{\beta_{jij}} \right) \pmb{\hat{g}}_{jij}^{\dagger} \pmb{v}_{jij} s_j [i]  \nonumber \\ 
	&\hspace{5mm}+ \sum\nolimits_{j=1}^L \sqrt{\dfrac{\rho_{\rm d}}{\lambda_j^{\rm zf}}} \sum\nolimits_{k=1, k \neq i}^K \left( \dfrac{\beta_{jil}}{\beta_{jij}} \right) \pmb{\hat{g}}_{jij}^{\dagger} \pmb{v}_{jkj} s_j [k] \nonumber \nonumber \\
	&\hspace{5mm}+ \sum\nolimits_{j=1}^L \sqrt{\dfrac{\rho_{\rm d}}{\lambda_j^{\rm zf}}} \sum\nolimits_{k=1}^K \pmb{\epsilon}_{jil}^{\dagger} \pmb{v}_{jkj} s_j [k] + w_{il}   \\ 
\nonumber	&\stackrel{(a)}{=}\! \underbrace{\sum_{j=1}^L \!\sqrt{\dfrac{\rho_{\rm d}}{\lambda_j^{\rm zf}}} \left( \dfrac{\beta_{jil}}{\beta_{jij}} \right)\!  s_j [i]}_{\textrm{ Desired signals}}   + \! \underbrace{\sum_{j=1}^L \! \sqrt{\dfrac{\rho_{\rm d}}{\lambda_j^{\rm zf}}} \sum_{k=1}^K \pmb{\epsilon}_{jil}^{\dagger} \pmb{v}_{jkj} s_j [k]}_{\textrm{ Interference due to channel estimation error }}\! + \!\underbrace{w_{il}}_{\rm Noise} \\[-2mm]
	&= \sum\nolimits_{j=1}^L \theta_{jij}^{\rm zf} s_j [i] + w_{il}^{\prime \prime} , \label{eq:zf:10:2}
	\end{align} }%
\endgroup where $(a)$ follows from the fact that $\pmb{v}_{jkj}^{\dagger} \pmb{\hat{g}}_{jmj} = \delta_{mk}$, $w_{il}^{\prime \prime}$ incorporates the last two terms after the third equality, and $\theta_{jij}^{\rm zf} = \sqrt{{\rho_{\rm d}}/{\lambda_j^{\rm zf}}} ( \beta_{jil} / \beta_{jij} )$.

We now have an $L$-user MAC in \eqref{eq:zf:10:2} associated with the estimate of $s_l [i]$. Following the treatment of the MRT case, and using standard random coding arguments as in \cite{el2011network}, it can be shown that the SD scheme achieves the rate tuple $\left( R_{i1}, ..., R_{iL} \right)$ that satisfies the set of conditions, obtained by replacing $y_{il}^{\rm mrt}$ in \eqref{eq:16} with the equivalent received signal $y_{il}^{\rm zf}$. The network-wide achievable rates are then obtained by taking the intersection of the achievable regions over all receivers.

Similar to the MRT case, the noise term in \eqref{eq:zf:10:2}  is neither Gaussian nor independent of the users' signals, nevertheless, it is uncorrelated from the desired signals. Thus, similar to those for MRT in Theorem~\ref{lemma:3}, we establish the following lower bounds for ZF using a worst-case uncorrelated noise technique.\vspace{-1mm}
\begin{theorem}\label{lemma:3:zf}
	For $\pmb{s}_l = \left[ s_l[1], s_l[2], ..., s_l[K] \right]^T \sim \mathcal{CN} \left(\pmb{0}, \; \pmb{I}_K  \right)$,  $l \in \lbrace 1, 2, ..., L \rbrace$, a set of lower bounds that can be achieved with ZF for the MAC given in \eqref{eq:zf:10:2}, associated with $\mathsf{U}_{il}$, is given as follows
	\begin{equation}\label{eq:21:1} 
	I \left(  y_{il}^{\rm zf} ; \; \pmb{s}_{\Omega} \; \Big\vert \; \pmb{s}_{\Omega^c} \right) \geq C \left( { P_1^{\rm zf} }/{\left(P_2^{\rm zf} + P_3^{\rm zf} \right)} \right), 
	\end{equation} where
	\par\nobreak \vspace*{-2mm} \begingroup \addtolength{\jot}{-0mm} { \begin{align} 
		P_1^{\rm zf} &= \sum\nolimits_{j \in \Omega}   \left( {\rho_{\rm d} \beta_{jil}^2}\right)/ \left({\lambda_j^{\rm zf} \beta_{jij}^2} \right) \label{eq:zf:14:1:1}\\
		P_2^{\rm zf} &= \sum\nolimits_{j=1}^L   \dfrac{\rho_{\rm d}}{\lambda_j^{\rm zf}}   \sum\nolimits_{k=1}^K   \dfrac{\beta_{jil} - \sqrt{\rho_{\rm p}} \beta_{jil} \alpha_{jil}}{(M-K) \sqrt{\rho_{\rm p}} \beta_{jkj} \alpha_{jkj}}, \label{eq:zf:14} \\
		 P_3^{\rm zf} &= 1. \label{eq:zf:14:1:2}
		\end{align}}%
	\endgroup
\end{theorem}    
The proof is similar to that of Theorem~\ref{lemma:3}, hence is omitted.

Similar to those for MRT, $P_1^{\rm zf}$ is the power of the desired signals associated with $\pmb{s}_{\Omega}$, while $P_2^{\rm zf}$ is the power of the interference due to the channel estimation error,  and $P_3^{\rm zf}$ is the noise power. Using \eqref{eq:zf:14:1:1}-\eqref{eq:zf:14:1:2}, the r.h.s of \eqref{eq:21:1} is given by  
 \par\nobreak \vspace{-1mm}  \addtolength{\jot}{-1mm} {\small
	\begin{align}
		C_{\rm LB}^{\rm zf} \!=  C \!\left( \frac{\sum\nolimits_{j \in \Omega} \left( {\rho_{\rm d} \beta_{jil}^2}\right)/ \left({\lambda_j^{\rm zf} \beta_{jij}^2} \right)   }{\sum\nolimits_{j=1}^L   \dfrac{\rho_{\rm d}}{\lambda_j^{\rm zf}}   \sum\nolimits_{k=1}^K   \dfrac{\beta_{jil} - \sqrt{\rho_{\rm p}} \beta_{jil} \alpha_{jil}}{(M-K) \sqrt{\rho_{\rm p}} \beta_{jkj} \alpha_{jkj}} +1} \right)\!\!. \!\!\!\label{eq:lb:zf} \vspace{1mm}
	\end{align}}%
Since $\lambda_j^{\rm zf} \!= \!\frac{1}{K (M-K)} \sum_{k=1}^K 1/ (\sqrt{\rho_{\rm p}} \beta_{jkj} \alpha_{jkj} )$, for fixed $K$ and large $M$, $\frac{1}{\lambda_j^{\rm zf}} \propto M$. Thus, $C_{\rm LB}^{\rm zf}$ grows as $\mathcal{O} (\log M)$. 
\vspace{-1mm}\subsection{Simultaneous Non-unique Decoding (SND)} \vspace{0mm}
It is important to note that $\mathsf{U}_{il}$ is only interested in correctly decoding $s_{l} [i]$. Thus, incorrectly decoding $s_j [i], \; j \neq l$, should not penalize the rates achievable at $\mathsf{U}_{il}$. Furthermore, the power of the received signal for the users located in distant cells is very small, and thus trying to decode signals of such users can reduce achievable rates considerably.
This motivates the need for \emph{non-unique} interference decoding (SND), which achieves higher rates than SD since the stringent condition of \emph{uniquely} decoding interference is now relaxed. 

In contrast to the SD scheme for which we focused on the MAC formed at $\mathsf{U}_{il}$, we now consider an $L$-user interference channel (IC) that consists of $L$ senders (\ied $\{\mathsf{BS}_l\}_{l=1}^L$) and $L$ receivers (\ied $\lbrace \mathsf{U}_{il} \rbrace_{l=1}^L$). In particular, $\mathsf{U}_{il}$ decodes only the intended message $s_l [i]$ uniquely, while the interference signals $s_j [i], \; j \neq l$ are decoded non-uniquely in the sense that incorrect decoding of the latter does not incur any penalty \cite[Chap.~6]{el2011network}. When point-to-point random codes are used, as in our setting, SND has been shown to achieve the capacity region of such $L$-user IC \cite{bandemer2015optimal}. In what follows, we only focus on the case of a 2-cell system for the sake of simplicity. From \cite{bandemer2015optimal}, the rate region for the 2-cell system achieved with SND at $\mathsf{U}_{i1}$ is given by \par\nobreak \vspace*{-3mm}  \addtolength{\jot}{1mm} { 
\begin{align}\label{snd:region} 
&R_{i1} \leq  I \left(  y_{i1}^{\textrm{mrt/zf}} ;  s_1 [i]  \big\vert \; s_2 [i] \right)  \\
\nonumber &R_{i1} \!  + \! \min \left\lbrace R_{i2} ,  I \left(  y_{i1}^{\textrm{mrt/zf}} ;   s_2 [i]   \big\vert   s_1 [i] \right) \right\rbrace  \!  \leq \!  I \left(  y_{i1}^{\textrm{mrt/zf}} ;   s_1 [i],   s_2 [i]   \right)\!.  
\end{align}}The $I(.;.|.)$ terms above can be readily lower bounded using Theorem~\ref{lemma:3} and Theorem~\ref{lemma:3:zf}, which scale as $\mathcal{O} (\log M)$ as in the SD case. Similarly, the rate region achieved at the other user, $\mathsf{U}_{i2}$, is obtained by replacing $y_{i1}^{\textrm{mrt/zf}}$ with $y_{i2}^{\textrm{mrt/zf}}$ and swapping appropriate indices in the above rate expressions. Hence, the network-wide rate region is given by the intersection of regions achieved at both users. Note that the SND region in \eqref{snd:region} can be viewed as the union of TIN and SD, i.e., TIN and SD regions are both strictly contained in SND, showing the advantages of non-uniquely decoding the interference. It should be pointed out that practical implementations of interference decoding in cellular networks using SND have been thoroughly studied in \cite{kim2015adaptive, kim2016interference}.  
\vspace{-1mm}\subsection{Partial Decoding (PD)} \vspace{0mm}
In SD and SND, interfering messages are decoded either uniquely or non-uniquely but in their \emph{entirety}. Nevertheless, it may be useful to decode only a \emph{part} of the message. Such flexible encoding/decoding is obtained by the celebrated Han-Kobayashi (HK) scheme \cite{hk}, which provides the best rates for the IC. While the HK scheme is complex in general, we characterize a PD scheme for our 2-cell setting based on the simplified rate splitting scheme of \cite{etw}. 

In SND, a message from $\mathsf{BS}_{l}$, $m_{il} \in \left[ 1 : 2^{n R_{il}} \right], l=1,2$, is encoded into a single codeword. In contrast, encoding in the PD scheme proceeds as follows:  message $m_{il}$ is first partitioned into two parts $m_{il}^{[1]}$ and $m_{il}^{[2]}$; part  $m_{il}^{[1]}$ is then encoded into  codeword $\pmb{s}_l^{[1]}[i](m_{il}^{[1]})$ of length $n$, while part  $m_{il}^{[2]}$ is encoded into another  codeword $\pmb{s}_l^{[2]}[i](m_{il}^{[1]}, m_{il}^{[2]})$ of length $n$; finally the latter codeword is superimposed (or layered) on the former to produce a single codeword for transmission $\pmb{s}_l[i](m_{il}^{[1]}, m_{il}^{[2]}) = \pmb{s}_l^{[1]}[i](m_{il}^{[1]}) + \pmb{s}_l^{[2]}[i](m_{il}^{[1]}, m_{il}^{[2]})$, $l = 1,2$. The total transmit power budget at $\mathsf{BS}_{l}$ is split into two fixed parts according to the power splitting variable $\mu_{il}  \in [0, 1]$: the fraction $\mu_{il}$ of the budget is allocated to the ``outer'' layer $\pmb{s}_l^{[2]}[i]$, while the fraction  $(1- \mu_{il})$ of the budget is allotted to the ``inner'' layer $\pmb{s}_l^{[1]}[i]$. Finally, $\pmb{s}_l^{[1]}[i]$ and $\pmb{s}_l^{[2]}[i]$ are chosen to be i.i.d., zero-mean circularly symmetric complex Gaussian, with powers determined by $\mu_{il}$.
 
Note that while full information $(m_{il}^{[1]},m_{il}^{[2]})$ is carried in $\pmb{s}_l[i]$, the inner layer $\pmb{s}_l^{[1]}[i]$ carries only coarse information $m_{il}^{[1]}$ which can be decoded without decoding $m_{il}^{[2]}$ in the outer layer. Hence, when the channel condition is poor, a receiver may only decode $m_{il}^{[1]}$, while for strong channels it may decode $(m_{il}^{[1]},m_{il}^{[2]})$ from both layers. Such adaptability helps PD to outperform SND, where messages are not layered and thus full interference is either non-uniquely decoded or treated as noise.
 
The decoding at $\mathsf{U}_{i1}$ proceeds as follows: it decodes both the inner and outer layers of the intended message $(m_{i1}^{[1]},m_{i1}^{[2]})$ uniquely, and only tries to non-uniquely decode the inner layer of the interfering message $m_{i2}^{[1]}$, while treating the outer layer $m_{i2}^{[2]}$ as noise. The other user, $\mathsf{U}_{i2}$, decodes similarly with the roles of the interfering and intended messages exchanged. The resulting rate region for the user pair $\mathsf{U}_{il}, l=1,2$, is given by $7$ inequalities found by changing the variables in \cite[Theorem~6.4]{el2011network} to those in our setting as follows: $X_1 \eqdef s_1[i], X_2 \eqdef s_2[i], U_1 \eqdef {s}_1^{[1]}[i], U_2 \eqdef {s}_2^{[1]}[i], Y_1 \eqdef  y_{i1}^{\textrm{mrt/zf}}, Y_2 \eqdef  y_{i2}^{\textrm{mrt/zf}}, Q \eqdef \emptyset$. 

Note that similar to those in the MRT/ZF cases for the SD scheme, the rate expressions involve $I(,;,|.)$ terms that cannot be computed exactly, hence we characterize lower bounds of these $I(,;,|.)$ terms following similar steps to those in Theorems~\ref{lemma:3} and \ref{lemma:3:zf}. Due to space limitations, we skip the expressions and detailed derivations, and instead, show that the achievable rates of PD outperform TIN/SD/SND, through simulation results presented in Section~V.

\vspace{-1mm}\section{Maximum Symmetric Rate Allocation} 
From \eqref{eq:19} and \eqref{eq:zf:10:2}, it is evident that users with relatively weak effective channel gain obtain smaller rates than those with a stronger channel. In order to ensure fairness among users, we study the maximum symmetric rate allocation problem for schemes TIN/SD/SND/PD, where the minimum achievable user-rate is maximized. The following analysis is for $\lbrace \mathsf{U}_{il} \rbrace_{l=1}^L $ across multiple cells that employ the same pilots. Since $i$ is arbitrary, the same analysis holds for other pilot-sharing user-sets.

When rate allocation is symmetric, all users are assigned the same rate. The maximum symmetric rate for a decoding scheme is given by $R_{\textrm{Sym}} = \max R$ such that the rate vector $\left( R, R, ..., R \right)$ belongs to the rate region achieved by that scheme at $\mathsf{U}_{il}$, for $l=1, ..., L$. Hence, the maximum symmetric rate vector $( R_{\textrm{Sym}}, R_{\textrm{Sym}},\ldots, R_{\textrm{Sym}})$ of a given scheme is found at the intersection of the diagonal $\left( R_{i1} = ... = R_{iL} \right)$ with the boundary of the achievable region obtained by that scheme.

\vspace{-1mm}\section{Simulation Results and Discussion} \vspace{-1mm}
To illustrate the performance of the different interference management schemes, TIN/SD/SND/PD, with symmetric spectral efficiency (SE in units of bits/sec/Hz) allocation, we simulate the downlink of a 2-cell massive MIMO system suffering from PC. In particular, we consider two hexagonal cells, as in Fig.~\ref{fig:two:cell}, with a radius of $r=400$ m, and $K=15$ users are uniformly distributed at random within the area of each cell, and at least $35$ m away from their BSs. To evaluate the performance, the average of maximum symmetric SEs is calculated over $150$ random realizations of user locations. The downlink transmit power of each BS is taken to be $40$ W (i.e., $46$ dBm), and to model large-scale fading coefficients, $\beta_{jkl}$, a path-loss model adopted from \cite{pathloss} is considered:\vspace{1mm}
\begin{align}\label{pathloss}
\left[ \beta_{jkl} \right]_{\rm dB} = &-13.54 -39.08 \log_{10} \left( d_{jkl}^{3D}  \right) \\
\nonumber &- 20 \log_{10} \left( f_c \right) + 0.6 \left(  h_{UT} - 1.5 \right), 
\end{align}
where $d_{jkl}^{3D}$ is the 3D distance (in meters) from $\mathsf{U}_{kl}$ to $\mathsf{BS}_j$, the carrier frequency is $f_c = 3.5$ GHz, $h_{UT} = 1.5$ m is the user height, and the BS height is $25$ m. Also, with a system bandwidth of $20$ MHz, the noise power is assumed to be $-101$ dBm. The SEs are calculated using the closed-form expressions of \eqref{eq:lb:mrt} and \eqref{eq:lb:zf} for MRT and ZF, respectively. 
\begin{figure}[t!]
\centering
\includegraphics[scale=0.39]{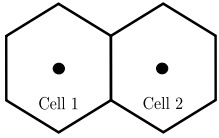}
\caption{A cellular system with two hexagonal cells, where BSs are located at the center of the cell and $K$ users are randomly positioned within the area of each cell. \label{fig:two:cell}\vspace{-7mm}}
\end{figure}
In Fig.~\ref{fig:1}, the achieved symmetric SEs of the different schemes TIN/SD/SND/PD are shown with MRT precoding, as a function of $M$. Note that the SEs for the PD scheme depend on the power splitting coefficients $\mu_{i1}, \mu_{i2} \in [0, 1]$, and hence the best SE for each value of $M$ is obtained by numerically optimizing $\mu_{i1}, \mu_{i2} \in [0, 1]$. First, notice that the SD scheme performs poorly compared to the other schemes, and unless $M$ becomes truly large (not shown in this figure) it can not outperform TIN. Moreover, for the range of $M$ shown, as $M$ increases, the SEs achieved by TIN, SND and PD improve, while the interference decoding schemes, SND and PD, outperform TIN for moderately large values of $M$ (i.e., $M \geq 128$). Moreover, since the PD scheme enables partial decoding of the PC interference with the remaining part treated as noise, this flexibility produces better rates than SND, such that PD outperforms both TIN and SND even for a small number of antennas $M$ (i.e., $M \geq 32$). More specifically, when using MRT, SND provides performance gain of about $8\%$ and $10\%$ over TIN for $M=128$ and $M=256$, respectively, while this gain improves to more than $17\%$ when $M=1024$. Furthermore, the gain offered by PD over TIN is about $29\%$ and $36\%$ for $M=128$ and $M=256$, respectively, and reaches approximately $72\%$ when $M=1024$. While not shown in the figure, as $M$ continues to grow and becomes truly large, the performance of SND and PD eventually converge, whereas the performance of TIN saturates to a constant value independent of $M$.\vspace{-0mm}
 
The performance of TIN/SD/SND/PD with ZF is demonstrated in Figs.~\ref{fig:2} and \ref{fig:3}. Fig.~\ref{fig:2} shows the achieved symmetric SEs with ZF for a range of moderately large values of $M$, while Fig.~\ref{fig:3} shows the same for a range of truly large $M$. While the latter covers a range of $M$ that is beyond practical, the results of Fig.~\ref{fig:3} can be used to confirm asymptotic performance limits as $M \rightarrow \infty$. One can notice in Fig.~\ref{fig:2} that all schemes achieve better SEs compared to MRT, which is due to the fact that ZF better mitigates multi-user interference. As such, compared to MRT, for the same number of antennas $M$, both SND and PD provide larger gains over TIN. More precisely, SND provides performance gains of $11\%$ and $17\%$ over TIN for $M=128$ and $M=256$, respectively, and this gain grows to more than $34\%$ when $M=1024$. On the other hand, due to the additional flexibility of partial interference decoding provided by PD, the gains offered by this scheme over TIN are about $96\%$ and $108\%$ for $M=128$ and $M=256$, respectively, while this gain improves to approximately $133\%$ when $M=1024$. Hence, these results confirm that the use of ZF precoding in conjunction with either SND or PD, can potentially be considered as an attractive solution to combat the PC interference in multi-cell massive MIMO systems. Particularly in the case of PD, the achieved SEs are increased by more than a factor of $2$ for a practical range of antennas $M$, which is a significant improvement over TIN. Similar to the MRT case, from Fig.~\ref{fig:2}, it is observed that for the practical range of $M$, SD is outperformed by the other schemes. The asymptotic limits presented in Fig.~\ref{fig:3} show that as $M$ continues to grow without limit, the SE obtained by TIN saturates, whereas the achieved SEs of both PD and SND monotonically increase with $M$ and will eventually converge. This is due to the fact that, as $M$ grows, and thus the strength of the PC interference increases, the power of the inner layer can be increased. This improvement is continued until the entire part of an interference signal can be decoded non-uniquely under PD (i.e., $\mu_{i1}, \mu_{i2} \rightarrow 0$), thus achieving a performance identical to that of SND. Lastly, Fig.~\ref{fig:3} shows that SD starts to outperform TIN at approximately $M \approx 4 \times 10^5$.\vspace{-0mm}      

Note that the SND scheme can be readily extended to more than two cells in downlink in a similar manner discussed in \cite{motlagh2019performance} for uplink. With small modifications, the message splitting technique for PD can also be extended to more than two cells. Simulation results (not shown here due to the space limitations) have revealed that doing so again results in significant gains over TIN and SND for the case of $L > 2$.
\begin{figure*}[t] 
	\captionsetup[subfloat]{captionskip=1mm}
	\centering    
	\subfloat[MRT]{\label{fig:1}\includegraphics[scale=0.45]{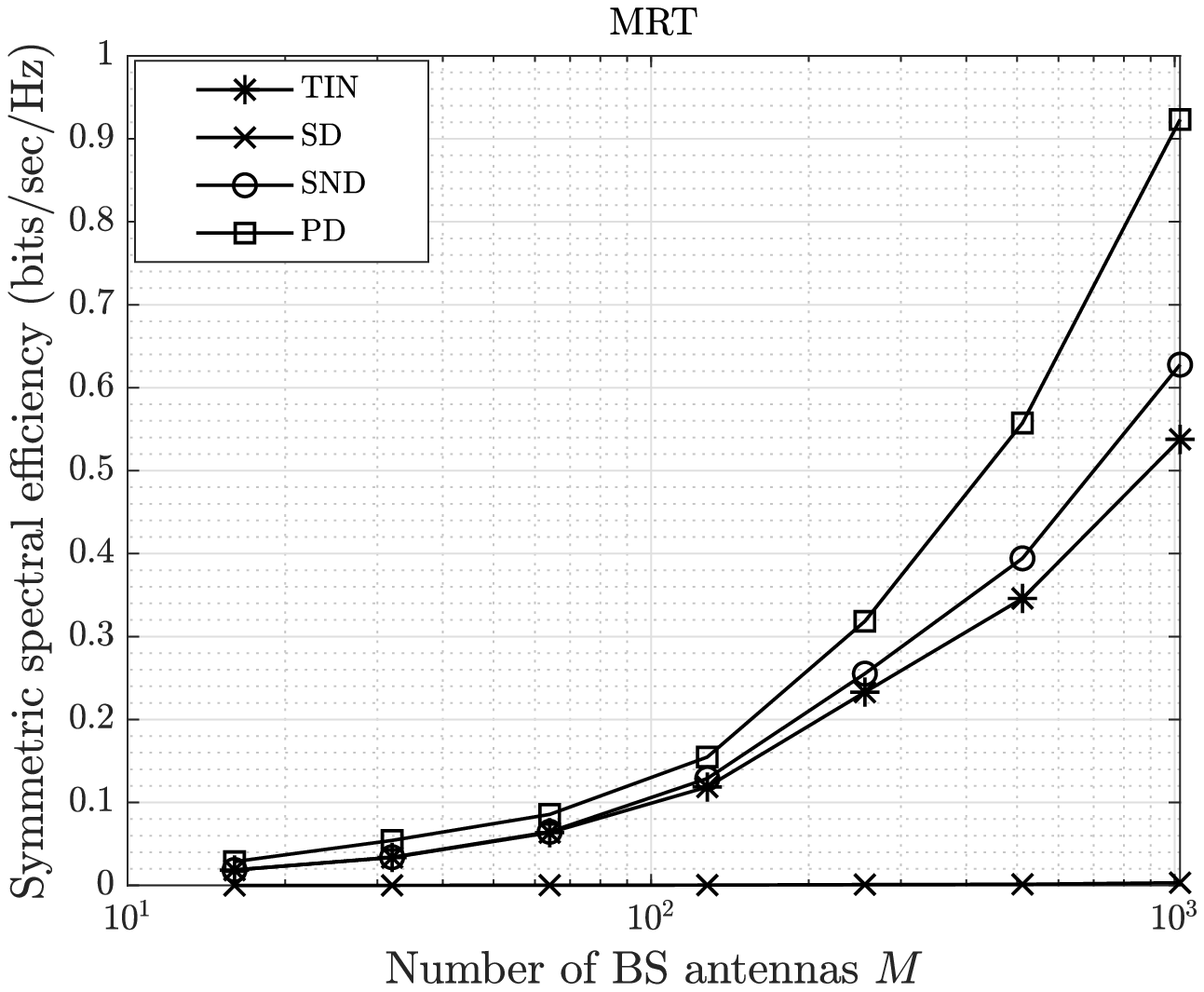}} \hspace*{-5mm}%
	\subfloat[ZF]{\label{fig:2}\includegraphics[scale=0.45]{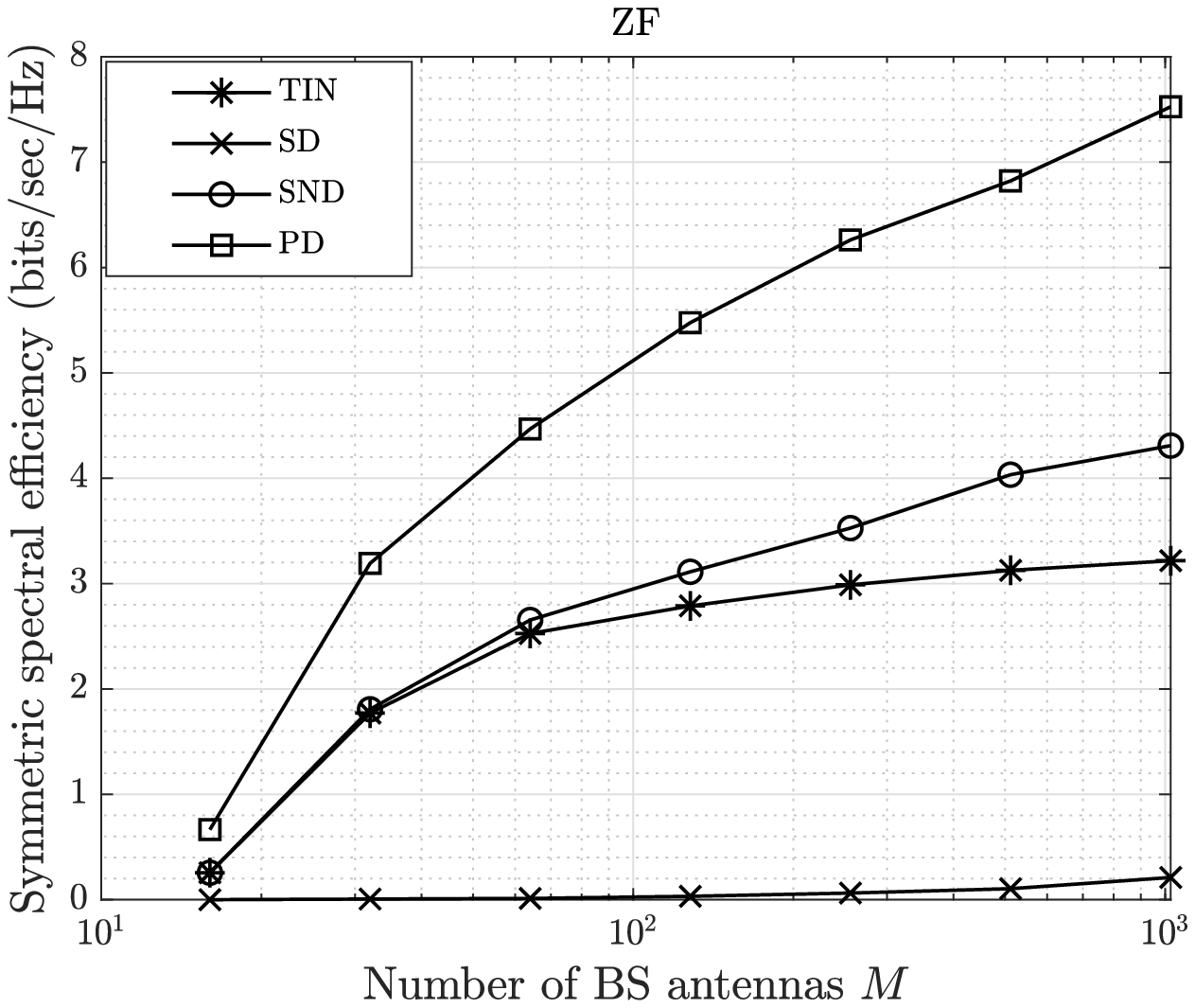}} \hspace*{-5mm}%
	\subfloat[ZF, asymptotic limits]{\label{fig:3}\includegraphics[scale=0.45]{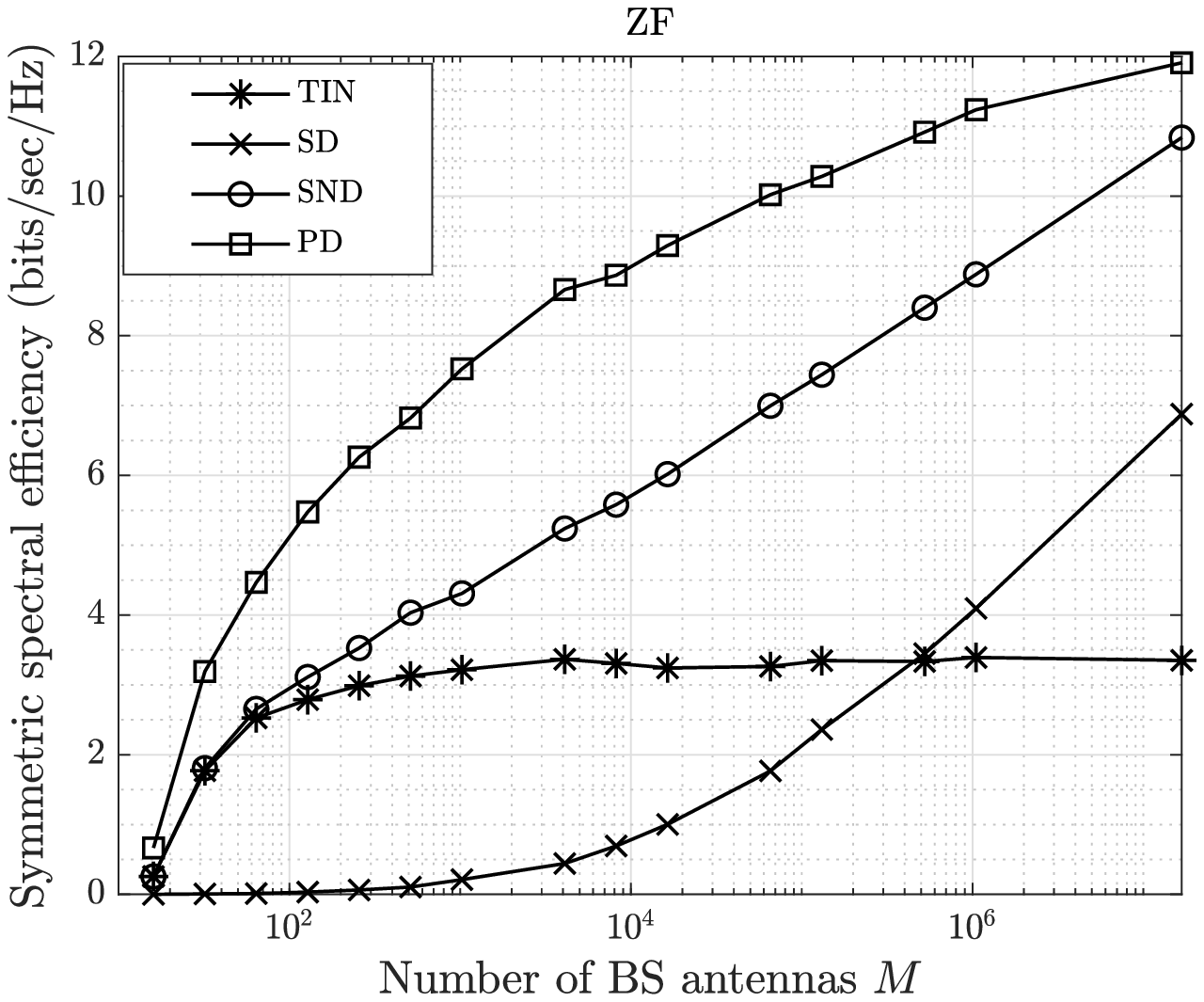}}\vspace*{-3mm}%
	\caption{\small (a) Symmetric SEs achieved with MRT under TIN/SD/SND/PD, (b) Symmetric SEs achieved with ZF under TIN/SD/SND/PD, for a range of moderately large $M$, (c) Asymptotic limits of symmetric SEs with ZF under TIN/SD/SND/PD for truly large values of $M$.}\vspace*{-3mm}
\end{figure*}
\vspace{-1mm}\section{Conclusion} \vspace{-0mm}
In this paper, for the downlink of a multi-cell massive MIMO system, schemes were proposed that decode the interference caused by PC, as opposed to treating it as noise. It was shown that for two linear precoding techniques, MRT and ZF, by decoding the PC interference, the achievable rates grow unbounded as $M \rightarrow \infty$. Furthermore, using a worst-case uncorrelated noise technique, new closed-form achievable rates for the interference decoding schemes with MRT and ZF were obtained which are shown to scale as $\mathcal{O} (\log M)$. Simulation results for the maximum symmetric SEs in a 2-cell system revealed that non-uniquely decoding the PC interference provides a better performance compared to TIN. Moreover, a partial interference decoding scheme based on rate splitting was proposed, outperforming \textit{all other schemes} for a practical range of $M$.  

\vspace*{-0mm}
\appendices
\vspace{-1mm}
\section{}
\vspace{-1mm}We employ the worst-case uncorrelated noise technique of \cite[Lemma~1]{motlagh2019performance}, \ied an achievable lower bound on the capacity of a MAC with uncorrelated additive non-Gaussian noise, $w^{\prime}_{il}$, is obtained by replacing the noise term with an independent zero-mean Gaussian noise having the same variance. 

We start by computing $P_1^{\rm mrt}$. Note that
\par\nobreak \vspace*{-2mm} \begingroup \addtolength{\jot}{-0mm} {\small \begin{align}\label{appB:1}
P_1^{\rm mrt} &= \sum\nolimits_{j \in \Omega}^L   {\rho_{\rm d}}/{\lambda_j^{\rm mrt}}   \left\vert  \mathbb{E} \left[ \hat{\pmb{g}}_{jil}^{\dagger} \hat{\pmb{g}}_{jij} \right] + \mathbb{E} \left[ \pmb{\epsilon}_{jil}^{\dagger} \hat{\pmb{g}}_{jij} \right] \right\vert^2  \notag \\
&\stackrel{\rm (a)}{=} \sum\nolimits_{j \in \Omega}^L    {\rho_{\rm d}}/{\lambda_j^{\rm mrt}}    \left\vert \mathbb{E} \left[ \hat{\pmb{g}}_{jil}^{\dagger} \hat{\pmb{g}}_{jij} \right]  \right\vert^2 \notag \\
&\stackrel{\rm (b)}{=} \sum\nolimits_{j \in \Omega}^L  \left( M \sqrt{\rho_{\rm p}} \beta_{jij} \alpha_{jij}  \right)^2   {\rho_{\rm d} \beta_{jil}^2}/ {\lambda_j^{\rm mrt} \beta_{jij}^2}       \notag \\
&= M^2 \sum\nolimits_{j \in \Omega}  \rho_{\rm d} \rho_{\textrm{p}}  \beta_{jil}^2  \alpha_{jij}^2 / \lambda_j^{\rm mrt}   ,  
\end{align}}%
where (a) is due to the fact that $\pmb{\epsilon}_{jil}$ and $\hat{\pmb{g}}_{jij}$ are independent and (b) follows from $\pmb{\hat{g}}_{jil} = (\beta_{jil} / \beta_{jij}) \pmb{\hat{g}}_{jij}$ and the distribution of $\pmb{\hat{g}}_{jkj}$ explained below \eqref{eq:6}. Since all three terms in the effective noise in \eqref{eq:19} are uncorrelated, we have  $\textrm{var} \left[ w_{il}^{\prime} \right] = P_2^{\rm mrt} + P_3^{\rm mrt} + P_4^{\rm mrt}. $ For the power of the interference due to beamforming gain uncertainty, $P_2^{\rm mrt}$, we write
\par\nobreak \vspace*{-2mm} \begingroup \addtolength{\jot}{-0mm} {\small\begin{align}
P_2^{\rm mrt} &= \sum\nolimits_{j=1}^L    {\rho_{\rm d}}/{\lambda_j^{\rm mrt}}  \;\; \mathbb{E}  \left[ \left\vert \hat{\pmb{g}}_{jil}^{\dagger} \hat{\pmb{g}}_{jij}  - \mathbb{E} \left[ \hat{\pmb{g}}_{jil}^{\dagger} \hat{\pmb{g}}_{jij}  \right] \right\vert^2 \right]  \nonumber  \\
&\hspace{4mm}+  \sum\nolimits_{j=1}^L    {\rho_{\rm d}} / {\lambda_j^{\rm mrt}}  \;\; \mathbb{E} \left[ \left\vert \pmb{\epsilon}_{jil}^{\dagger} \hat{\pmb{g}}_{jij} \right\vert^2 \right]. \label{appB:4}
\end{align}}%
For the first term in \eqref{appB:4}, it is obtained
\par\nobreak \vspace*{-2mm} \begingroup \addtolength{\jot}{-0mm} {\small\begin{align}
&\sum\nolimits_{j=1}^L   \dfrac{\rho_{\rm d}}{\lambda_j^{\rm mrt}}   \mathbb{E}  \left[ \left\vert \hat{\pmb{g}}_{jil}^{\dagger} \hat{\pmb{g}}_{jij}  - \mathbb{E} \left[ \hat{\pmb{g}}_{jil}^{\dagger} \hat{\pmb{g}}_{jij}  \right] \right\vert^2 \right] \nonumber \\
&= \sum\nolimits_{j=1}^L   \dfrac{\rho_{\rm d}}{\lambda_j^{\rm mrt}}   \textrm{var} \left[ \hat{\pmb{g}}_{jil}^{\dagger} \hat{\pmb{g}}_{jij} \right]  = M \sum\nolimits_{j=1}^L   \dfrac{\rho_{\rm d}}{\lambda_j^{\rm mrt}}   \rho_{\rm p} \beta_{jil}^2  \alpha_{jij}^2. \label{appB:6}
\end{align}}%
Similarly, for the second term in \eqref{appB:4}, we obtain
\par\nobreak \vspace*{-2mm} \begingroup \addtolength{\jot}{0mm} {\small\begin{align}
&\sum_{j=1}^L   \dfrac{\rho_{\rm d}}{\lambda_j^{\rm mrt}}   \mathbb{E} \left[ \left\vert \pmb{\epsilon}_{jil}^{\dagger} \hat{\pmb{g}}_{jij} \right\vert^2 \right]  = \sum_{j=1}^L  \dfrac{\rho_{\rm d}}{\lambda_j^{\rm mrt}}   \mathbb{E} \left[   \hat{\pmb{g}}_{jij}^{\dagger} \pmb{\epsilon}_{jil} \pmb{\epsilon}_{jil}^{\dagger} \hat{\pmb{g}}_{jij}   \right] \notag \\
 &= M \sum\nolimits_{j=1}^L   \rho_{\rm d}  \beta_{jil} \left( 1 - \sqrt{\rho_{\rm p}}  \alpha_{jil} \right) \left( \sqrt{\rho_{\rm p}} \beta_{jij} \alpha_{jij} \right) / {\lambda_j^{\rm mrt}}. \label{appB:7}
\end{align}}%
Using \eqref{appB:6} and \eqref{appB:7}, $P_2^{\rm mrt}$ is obtained as in \eqref{eq:22:1}.
Next,  
\par\nobreak \vspace*{-2mm} \begingroup \addtolength{\jot}{0mm} {\small\begin{align}
&P_3^{\rm mrt} = \sum\nolimits_{j=1}^L \hspace{-1mm} \dfrac{\rho_{\rm d}}{\lambda_j^{\rm mrt}}    \sum\nolimits_{k=1, k \neq i}^K  \hspace{-1mm} \left(   \mathbb{E} \left[ \left\vert \hat{\pmb{g}}_{jil}^{\dagger}  \hat{\pmb{g}}_{jkj} \right\vert^2 \right] \! + \! \mathbb{E} \left[ \left\vert \pmb{\epsilon}_{jil}^{\dagger} \hat{\pmb{g}}_{jkj} \right\vert^2  \right] \right)   \notag \\
&=  \sum\nolimits_{j=1}^L    {\rho_{\rm d}} / {\lambda_j^{\rm mrt}}   \sum\nolimits_{k=1, k \neq i}^K  \Big(  M \left( \sqrt{\rho_{\rm p}} \beta_{jil} \alpha_{jil} \right) \left( \sqrt{\rho_{\rm p}} \beta_{jkj} \alpha_{jkj} \right)  \nonumber \\
 &\hspace{4mm}+  M \left( \sqrt{\rho_{\rm p}} \beta_{jkj} \alpha_{jkj} \right) \left( \beta_{jil} - \sqrt{\rho_{\rm p}} \beta_{jil} \alpha_{jil} \right)  \Big) \notag \\
 &=  M  \sum\nolimits_{j=1}^L   {\rho_{\rm d}} / {\lambda_j^{\rm mrt}}    \sum\nolimits_{k=1, k \neq i}^K  \sqrt{\rho_{\rm p}} \beta_{jkj} \alpha_{jkj}  \beta_{jil} . \notag
\end{align}}%
Finally, the power of the noise $w_{il}$ in \eqref{eq:18} is $P_4^{\rm mrt}\!=\!1$.

\vspace{-1mm}
\renewcommand{\baselinestretch}{0.82}
\bibliographystyle{IEEEtran} 
\bibliography{IEEEabrv,RefTCOM}
\end{document}